\documentclass[onecolumn]{IEEEtran}
\IEEEoverridecommandlockouts
\usepackage{cite}
\usepackage{amsmath}
\usepackage{amsfonts}
\usepackage{pifont}
\usepackage{amssymb}
\usepackage{amsthm}
\usepackage{tikz}
\usepackage{cases}
\usepackage{graphicx}
\usepackage{float,stfloats}
    \graphicspath{{../}}
    \DeclareGraphicsExtensions{.pdf}
\usepackage[justification=centering]{caption}
\usepackage[caption=false,font=footnotesize]{subfig}
\usepackage{multirow}
\usepackage{booktabs}
\usepackage{url}
\usepackage{xtab}
\usepackage{tabu}
\usepackage{longtable}
\usepackage{algorithm}
\usepackage{algorithmic}
\usepackage{enumerate}
\usepackage{makecell}
\usepackage{lipsum}
\usepackage{multicol}
\usepackage{mathdots}
\usepackage{extarrows}
\usepackage{color,xcolor}
\usepackage{varwidth}
\usepackage{bm} 
\usetikzlibrary{arrows}
\newsavebox{\mybox}
\theoremstyle{plain}

\newtheorem{theorem}{Theorem}
\newtheorem{lemma}[theorem]{Lemma}

\theoremstyle{definition}
\newtheorem{example}{Example}

\newtheorem{remark}{Remark}

\def\BibTeX{{\rm B\kern-.05em{\sc i\kern-.025em b}\kern-.08em
    T\kern-.1667em\lower.7ex\hbox{E}\kern-.125emX}}

\begin{document}
\title{Explicit Construction of Minimum Bandwidth Rack-Aware Regenerating Codes
}

\author{\IEEEauthorblockN{Liyang Zhou, ~Zhifang Zhang}\\
\IEEEauthorblockA{\fontsize{9.8}{12}\selectfont KLMM, Academy of Mathematics and Systems Science, Chinese Academy of Sciences, Beijing 100190, China\\
School of Mathematical Sciences, University of Chinese Academy of Sciences, Beijing 100049, China\\
Emails: zhouliyang17@mails.ucas.ac.cn,~~zfz@amss.ac.cn}

}\maketitle

\thispagestyle{empty}

\begin{abstract}
In large data centers, storage nodes are organized in racks, and the cross-rack transmission dominates the bandwidth cost. For the repair of single node failures, codes achieving the tradeoff between the storage redundancy and cross-rack repair bandwidth are called rack-aware regenerating codes (RRCs). In this work, we give the first explicit construction of RRCs with the minimum repair bandwidth (i.e., the cross-rack bandwidth equals the storage size per node). Our construction applies to all admissible parameters and has the lowest sub-packetization level. Moreover, the underlying finite fields are of size comparable to the number of storage nodes, which makes our codes more implementable in practice. Finally, for the convenience of practical use, we also establish a transformation to convert our codes into systematic codes.
\end{abstract}

\maketitle{}

\section{Introduction}\label{sec0}
Erasure codes are increasingly adopted in modern storage systems (such as Windows Azure Storage \cite{Huang2012WAS}, Facebook storage \cite{Facebook2013}, etc.) to ensure fault-tolerant storage with low redundancy. Meanwhile, efficient repair of node failures becomes an important issue in coding for distributed storage. An important metric for the repair efficiency is the repair bandwidth, i.e., the total number of symbols downloaded from the helper nodes for recovery of the data on failed nodes. The celebrated work \cite{Dimakis2011} initiated the study of regenerating codes that can minimize the repair bandwidth for given storage redundancy. By now fruitful results have been achieved in regenerating codes which are partly included in the survey \cite{Balaji2018Survey}.

The initial model for regenerating codes treats the bandwidth cost equally between all storage nodes. For simplicity, we call this model as the homogeneous model throughout. However, large data centers usually possess heterogeneous structure where the storage nodes are organized in racks (or clusters), so it permits to differentiate between communication within the rack (or cluster) and the cross-rack (or inter-cluster) downloads. Specifically, let $n=\bar{n}u$ and the $n$ nodes are organized in $\bar{n}$ racks each containing $u$ nodes. A data file consisting of $B$ symbols is stored across the $n$ nodes each storing $\alpha$ symbols such that any $k=\bar{k}u+u_0~(0\leq u_0<u)$ nodes can retrieve the data file. Suppose a node fails. A self-sustaining system is to generate a replacement node by downloading data from surviving nodes. We call the rack containing the failed node as the {\it host rack}. The node repair is based on the two kinds of communication below:
\begin{enumerate}
	\item \textbf{Intra-rack transmission.}
	All surviving nodes in the host rack transmit information to the replacement node.
	\item \textbf{Cross-rack transmission.}
	Outside the host rack, $\bar{d}$ helper racks each transmit $\beta$ symbols to the replacement node.
\end{enumerate}
Particularly when $u=1$, it degenerates into the homogenous model where regenerating codes have been well studied. Here we focus on the case $u>1$. In general, the cross-rack communication cost is much more expensive than the intra-rack communication cost, so it is reasonable to assume the nodes within each rack can communicate freely without taxing the system bandwidth. Consequently, the $\beta$ symbols provided by each helper rack are computed from the data stored in all nodes in that helper rack, and the repair bandwidth $\gamma$ only dependents on the cross-rack transmission, i.e., $\gamma=\bar{d}\beta$.

The above model of rack-aware storage was introduced in \cite{Hu} for studying the optimal node repair in hierarchical data centers. More specifically, the authors derived a lower bound on the repair bandwidth for codes with the minimum storage and also presented an existential construction of codes attaining this lower bound. Then, Hou et al. \cite{Hou} further explored the rack-aware storage model, both for the minimum-storage and minimum-bandwidth scenarios. Similarly to \cite{Dimakis2011}, Hou et al. firstly derived a cut-set bound for the rack-aware model and obtained a tradeoff between $\alpha$ and $\gamma$. The codes with parameters lying on the tradeoff curve are called rack-aware regenerating codes (RRCs). Two extreme points on the tradeoff curve respectively correspond to the minimum storage rack-aware regenerating (MSRR) code and minimum bandwidth rack-aware regenerating (MBRR) code. They also gave existential constructions of MSRR and MBRR codes over sufficiently large fields. Shortly after that, Chen and Barg \cite{Chen} extended the construction of MSR codes in the homogeneous model \cite{Ye2016sub-} to obtain the first explicit constructions of MSRR codes for all admissible parameters. Recently, the sub-packetization level of MSRR codes was further reduced \cite{Hou2020,Zhou2020}. However, explicit constructions of MBRR codes remain unsolved.

Meanwhile, some varieties of rack-aware regenerating codes were studied. In \cite{Prakash2018} the authors characterized the optimal repair bandwidth for the clustered storage system where the data reconstruction is realized by connecting to any $\bar{k}$ clusters rather than any $k$ nodes. The authors of \cite{Sohn2019} assumed a fixed ratio of the cross-cluster to intra-cluster bandwidth and  defined the repair bandwidth as the sum of cross-cluster and intra-cluster repair bandwidth. They constructed explicit codes \cite{Sohn2019arX} for minimizing the repair bandwidth when all the remaining nodes participate in the repair of one node failure. The paper \cite{Gupta} extended the rack-aware regenerating codes to the cooperative repair model with the assumption that the multiple node failures are evenly distributed in multiple racks.

\subsection{Our contribution}
In this work we focus on the rack-aware regenerating codes with the minimum repair bandwidth, i.e., the MBRR codes. According to \cite{Hou}, the MBRR code has parameters
\begin{equation}\label{bw-bound1}
\alpha=\bar{d}\beta=B\bar{d}\big/\big(k\bar{d}-\frac{\bar{k}(\bar{k}-1)}{2}\big)
\end{equation}
where $\bar{k}=\lfloor\frac{k}{u}\rfloor$. The authors of \cite{Hou} firstly constructed MBRR codes over the field of size larger than $ B{\sum_{i=1}^{\min\{k,\bar{n}\}}}\binom{n-\bar{n}}{k-i}\binom{\bar{n}}{i}$.
In their construction, a product-matrix structure of the MBR codes in \cite{Kumar2011} was used to ensure the optimal repair bandwidth, while the data reconstruction from any $k$ nodes was guaranteed by the invertibility of some related matrices. However, they failed to give explicit constructions of the corresponding matrices.
We conquered this problem through a nice merge of the multiplicative subgroup design into the product-matrix MBR codes. The multiplicative subgroup design was first adopted in \cite{LRCFamily2014} for ensuring local repairability in generating linear codes with the optimal distance, and then in \cite{Chen} for constructing MSRR codes from the parity-check matrix. Here we apply the design in the product-matrix structure, achieving the minimum bandwidth repair and $k$-reconstruction simultaneously in an explicit way.


Therefore, we obtain the first explicit construction of MBRR codes for all the admissible parameters considered in \cite{Hou} (i.e., $\bar{d}\geq \bar{k}$).  The codes are built over finite fields of size comparable to $n$ which is much smaller than that required in \cite{Hou}. Moreover,
our codes apply to the scalar case $\beta=1$ and thus have the smallest sub-packetization level $\alpha=\bar{d}$. Finally, for the convenience of practical use, we also establish a transformation to convert our codes into systematic codes.

Comparing with minimum bandwidth regenerating (MBR) codes in other storage models, the MBRR codes built in this paper have advantages in storage overhead, fault tolerance and some other aspects. In more detail,
\begin{enumerate}
 \item A drawback of the MBR code in the homogeneous storage model \cite{Dimakis2011} is that its storage overhead is always greater than $2$ which means the redundancy must have size at least as large as the data file. However, this restriction no longer exists for the MBRR codes. For example, let $u=5$, $n-k=6$ and $\bar{d}=\bar{n}-1$. When $\bar{n}\leq 51$, our MBRR codes can be built over GF$(2^8)$ with storage overhead $1.22$ as $\bar{n}=10$ and $1.13$ as $\bar{n}=40$.

 \item The clustered storage model considered in \cite{Prakash2018} requires the data reconstruction is achievable by connecting to any $\bar{k}$ clusters (i.e., racks) rather than any $k$ nodes. As a result, the regenerating codes in this model only have fault tolerance $\bar{n}-\bar{k}$. Namely, when $\bar{n}-\bar{k}+1$ node failures spread in distinct clusters, the data file gets lost. The value of $\bar{n}-\bar{k}$ is quite small if low storage overhead is desired. In contrast, our MBRR codes have fault tolerance $n-k\approx(\bar{n}-\bar{k})u$ and storage overhead near to $\bar{n}/\bar{k}$ for properly chosen parameters.

 \item The clustered storage model considered in \cite{Sohn2019} assumes a fixed ratio $\epsilon$ of the intra-cluster to cross-cluster bandwidth and  defined the repair bandwidth as the sum of cross-cluster and intra-cluster repair bandwidth. When $\epsilon=1$ it degenerates into the homogeneous storage model. When $\epsilon=0$, however, it differentiates from our rack-aware storage model in that no communication within each helper rack is considered, namely, each node in the helper racks uploads data directly to the replacement node, while in the rack-aware model a centralized processing is implemented within each helper rack before transmitting the help data. Therefore, the model in \cite{Sohn2019} tends to need larger cross-rack repair bandwidth than our rack-aware model. Moreover, for general values of $\epsilon$ ($0\leq \epsilon\leq 1$), it is difficult to design codes for minimizing the repair bandwidth. In \cite{Sohn2019arX} the authors constructed the MBR and MSR codes only in the very special case that all the remaining $n-1$ nodes participate in the repair of a single node failure. In contrast, our codes allows the number of helper racks $\bar{d}$ ranging from $\bar{k}$ to $\bar{n}-1$, introducing more flexibility in the repair process.

\end{enumerate}

The remaining of the paper is organized as follows. Section II presents the explicit construction of MBRR codes. Section III describes the transformation to systematic MBRR codes. Section IV concludes the paper.

\section{Explicit construction}\label{sec1}
First we introduce some notations. For nonnegative integers $m<n$, let $[n]=\{1,...,n\}$ and $[m,n]=\{m,m+1,...,n\}$. Throughout the paper we label the racks from $0$ to $\bar{n}-1$ and the nodes within each rack from $0$ to $u\!-\!1$. Moreover, we represent each node by a pair $(e,g)\in[0,\bar{n}\!-\!1]\times [0,u\!-\!1]$ where $e$ is the rack index and $g$ is the node index within the rack. Our MBRR codes are built over a finite field $F$ satisfying $u|(|F|-1)$ and $|F|>n$. The codes apply to the scalar case $\beta=1$. Therefore, the data file to be stored across $n$ nodes is composed of $B$ symbols from $F$ and each node stores $\bar{d}$ symbols.

Suppose $\beta=1$. According to (\ref{bw-bound1}) the MBRR codes have
 \begin{equation}\label{eqScalarMBRR}\alpha=\bar{d},~~~B=k\bar{d}-\frac{\bar{k}(\bar{k}-1)}{2}=(k\!-\!\bar{k})\bar{d}+\frac{\bar{k}(\bar{k}+1)}{2}+\bar{k}(\bar{d}\!-\!\bar{k})\;.\end{equation}
Next we describe the construction of MBRR codes in three steps.

{\bf Step 1.}
Define two sets $J_1=\{tu+u-1: t\in[0,\bar{d}-1]\}$ and $J_2=[0,k-1]-J_1$. It can be seen that $|J_1|=\bar{d}$ and $|J_2|=k-\bar{k}$. Moreover, denote $J=J_1\cup J_2$ which can be rewritten as \begin{equation}J=[0,k-1]\cup \{tu+u-1: t\in[\bar{k},\bar{d}-1]\}\;.\label{eqJ}\end{equation}

Then we put the $B$ symbols of the data file into a $\bar{d}\times(k-\bar{k}+\bar{d})$ matrix $M=(m_{i,j})_{i\in[0,\bar{d}-1],j\in J}$. Note the columns of $M$ are indexed by the set $J$. Let $M_1$ and $M_2$ denote the submatrices of $M$ restricted to the columns indexed by $J_1$ and $J_2$ respectively. Moreover, $M_1$ has the following form
\begin{equation}\label{ms2}
M_1=\begin{pmatrix}
S&T\\
T^\tau&0	
\end{pmatrix},	
\end{equation}
where $S$ is a $\bar{k}\times\bar{k}$ symmetric matrix with the upper-triangular half filled up with $\frac{\bar{k}(\bar{k}+1)}{2}$ symbols of the data file, $T$ is a $\bar{k}\times(\bar{d}-\bar{k})$ matrix filled up with $\bar{k}(\bar{d}-\bar{k})$ symbols of the data file, and $T^\tau$ denotes the transpose of $T$. As a result, $M_1$ is a $\bar{d}\times\bar{d}$ symmetric matrix containing $\frac{\bar{k}(\bar{k}+1)}{2}+\bar{k}(\bar{d}\!-\!\bar{k})$ symbols altogether. The matrix $M_2$ is a $\bar{d}\times(k-\bar{k})$ matrix containing the remaining $(k-\bar{k})\bar{d}$ symbols of the data file.
The construction of the matrix $M$ is displayed in Fig. \ref{fg1}.

\begin{figure}[ht]
\begin{center}
\includegraphics[width=0.7\columnwidth]{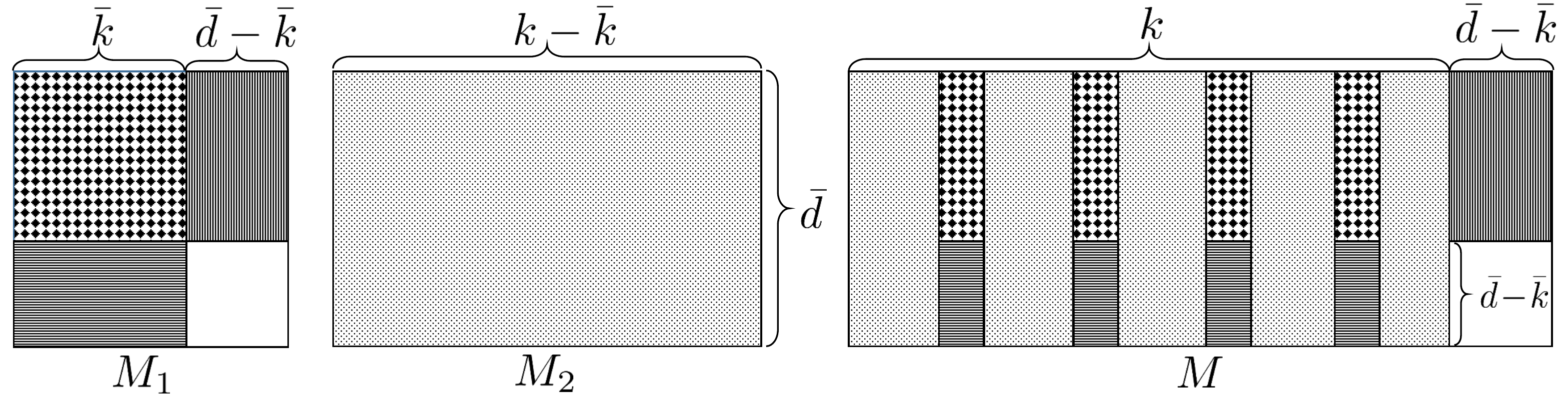}
\end{center}
\caption{Construction of the matrix $M$.}\label{fg1}
\end{figure}

{\bf Step 2.} For $i\in[0,\bar{d}-1]$ define polynomials
\begin{equation}f_i(x)\!=\!\sum_{j\in J}m_{i,j}x^j\!=\!\sum_{j=0}^{k-1}m_{i,j}x^j\!+\!\sum_{t=\bar{k}}^{\bar{d}-1}m_{i, tu\!+\!u\!-\!1}x^{tu\!+\!u\!-\!1}\label{polys}\;,\end{equation} where the second equality comes from the form of $J$ in (\ref{eqJ}).

Then each codeword of the MBRR code is composed of evaluations of the $\bar{d}$ polynomials at specially chosen points. Before going to Step 3, we use an example to further illustrate the first two steps.

\begin{example}\label{eg1}
Suppose $n=12, k=7, u=3, \bar{d}=3$ and $\beta=1$. According to \eqref{bw-bound1} the MBRR codes have
$\alpha=3, B=20$.	
In order to store $B=20$ data symbols $s_1,s_2,...,s_{20}$, in Step 1 we define a $3\times 8$ matrix $M$ as follows:
$$M\!=\!\begin{pmatrix}s_1&s_4&{\bm s_7}&s_{10}&s_{13}&{\bm s_8}&s_{18}&{\bm s_9}\\
s_2&s_5&{\bm s_8}&s_{11}&s_{14}&{\bm s_{16}}&s_{19}&{\bm s_{17}}\\
s_3&s_6&{\bm s_9}&s_{12}&s_{15}&{\bm s_{17}}&s_{20}&{\bm 0}
\end{pmatrix}\;,$$
\noindent where the symbols in bold face form the matrix $M_1$. Note $M_1$ is symmetric and $M$ contains redundant symbols due to the symmetry. Then in Step 2, we define three polynomials $f_0(x),f_1(x),f_2(x)$ such that
$$\begin{pmatrix}
  f_0(x)\\f_1(x)\\f_2(x)\end{pmatrix}=M\begin{pmatrix}
    1\\x\\\vdots\\x^{6}\\x^8
  \end{pmatrix}\;.
$$
\end{example}

{\bf Step 3. }
Let $\xi$ be a primitive element of $F$ and $\eta$ be an element of $F$ with multiplicative order $u$. Note $\xi$ and $\eta$ are fixed and publicly known. Then denote $\lambda_{(e,g)}=\xi^e\eta^g$ for $e\!\in\![0,\bar{n}\!-\!1],g\!\in\![0,u\!-\!1]$. It follows $\lambda_{(e,g)}\neq\lambda_{(e',g')}$ for $(e,g)\neq (e',g')$ because $(\eta^{g'-g})^u=1$ while $(\xi^{e-e'})^u\neq 1$ for $e\!\neq\! e'\!\in\![0,\bar{n}-1]$. That is, we select $n$ distinct elements $\lambda_{(e,g)}\in F$, where $(e,g)\in[0,\bar{n}-1]\times[0,u-1]$.
Finally we construct a code $C$ by letting the node $(e,g)$ store the $\bar{d}$ symbols: $f_0(\lambda_{(e,g)}),f_1(\lambda_{(e,g)}),...,f_{\bar{d}-1}(\lambda_{(e,g)})$.

In summary, the code $C$ can be expressed as
\begin{equation}\label{cons}C=M\Lambda\end{equation}
where $M$ is the $\bar{d}\times(k-\bar{k}+\bar{d})$ matrix in step 1, $\Lambda=(\lambda_{(e,g)}^j)$ with row index $j\in J$ and column index $(e,g)\in[0,\bar{n}-1]\times[0,u-1]$, and $C$ is a $\bar{d}\times n$ code matrix such that each node stores a column of $C$.

\vspace{4pt}
\begin{remark}
The code construction (\ref{cons}) coincides with the product-matrix framework proposed in \cite{Kumar2011}.
For simplicity, we call $M$ the message matrix, $\Lambda$ the encoding matrix and $C$ the code matrix.
Moreover,  a sub-matrix of $M$, i.e., $M_1$, has exactly the same form as the message matrix of the product-matrix MBR codes. The matrix $\Lambda$ uses the same multiplicative subgroup as in \cite{Chen}. Our trick mostly relies on the design of the set $J$, namely, the way of inserting $M_1$ into the remaining message matrix, which enables the code satisfy the minimum cross-rack repair bandwidth as well as the data reconstruction from any $k$ nodes.
\end{remark}

In the following we prove the code $C$ constructed in (\ref{cons}) is an MBRR code by showing it satisfies the data reconstruction and optimal repair property.

\subsection{Data Reconstruction}
We are to prove any $k$ nodes can together recover the data file. It is obvious that the data file is retrieved iff all entries of the matrix $M$ are recovered. We have the following theorem.

\begin{theorem}\label{thm1}
Given the code $C$  in (\ref{cons}), the matrix $M$ can be recovered from any $k$ columns of $C$.
\end{theorem}

\begin{proof}
We equivalently prove the $\bar{d}$ polynomials defined in (\ref{polys}) can be recovered from any $k$ columns of $C$ because coefficients of these polynomials contain exactly all entries of $M$.

First, from the construction of $M$ as displayed in Fig \ref{fg1} one can see for $i\in[\bar{k},\bar{d}-1]$, the polynomial $f_i(x)$ has degree at most $k-1$. Thus from any $k$ columns (actually, just the bottom $\bar{d}-\bar{k}$ rows of these columns) of $C$ one can obtain evaluations of these $\bar{d}-\bar{k}$ polynomials at $k$ distinct points, and then recover the polynomials by Lagrange interpolation.  After that, by the symmetric structure of $M_1$, coefficients of the terms of degree greater than $k-1$ of the first $\bar{k}$ polynomials, i.e., $f_i(x)$ for $i\in[0,\bar{k}\!-\!1]$, are simultaneously obtained from the last $\bar{d}-\bar{k}$ polynomials. Thus recovery of the first $\bar{k}$ polynomials reduces to interpolating $\bar{k}$ polynomials of degree at most $k-1$. Therefore, one can recover the remaining $k$ coefficients of each of the first $\bar{k}$ polynomials from the $k$ columns of $C$.
\end{proof}

\begin{remark}
Actually, the explicit design that enables data reconstruction from any $k$ nodes is the most challenging part.
A natural idea is introducing an invertible sub-matrix in any $k$ columns of the encoding matrix, like the Reed-Solomon code satisfying the MDS property. However, this is not easy if the optimal repair property is simultaneously in concern. The work of \cite{Hou} only proved existence of these invertible sub-matrices over sufficiently large fields. Our trick is constructing some polynomials of degree at most $k-1$ whose coefficients are partly projected to the terms of higher degree of the remaining polynomials. Moreover, reuse of these coefficients enables optimal repair of single node failures.
\end{remark}

\subsection{Node Repair}
We are to prove any single node erasure can be recovered from all the remaining $u-1$ nodes in the same rack, as well as $\bar{d}$ helper racks each transmitting $\beta=1$ symbol.
In short, the node repair of our MBRR code relies on the node repair of the product-matrix MBR code \cite{Kumar2011} and the local repair property within each rack.

Firstly we illustrate the local repair property. For each rack $e\in[0,\bar{n}-1]$ and $i\in[0,\bar{d}-1]$, define a polynomial $h^{(e)}_i(x)=\sum_{j=0}^{u-1}h^{(e)}_{i,j}x^j$ where
\begin{equation}\label{poly}
h_{i,j}^{(e)}=\begin{cases}
\sum_{t=0}^{\bar{k}}m_{i,tu+j}\cdot\xi^{etu}\ \ \ \ \ \ \ \mathrm{if} \ 0\leq j<u_0\\
\sum_{t=0}^{\bar{k}-1}m_{i,tu+j}\cdot\xi^{etu}\ \ \ \ \ \ \ \mathrm{if} \ u_0\leq j<u-1\\
\sum_{t=0}^{\bar{d}-1}m_{i,tu+u-1}\cdot\xi^{etu}\ \ \ \ \mathrm{if} \ j=u-1
\end{cases}\;.
\end{equation}

\begin{lemma}\label{lemma} For all $e\in[0,\bar{n}-1]$ and $i\in[0,\bar{d}-1]$, it holds $f_i(\lambda_{(e,g)})=h^{(e)}_i(\lambda_{(e,g)})$ for all $g\in[0,u-1]$.
\end{lemma}
\begin{proof}
From \eqref{polys} it has $f_i(\lambda_{(e,g)})=\sum_{j\in J}m_{i,j}\lambda_{(e,g)}^j$. Next we rearrange the terms of $f_i(\lambda_{(e,g)})$ according to the values of $j\ {\rm mod}\ u$ for all $j\in J$. Specifically, denote $j=tu+\nu$, where $\nu= j\ {\rm mod}\ u$. Then
\begin{equation*}
f_i(\lambda_{(e,g)})=\sum_{\nu=0}^{u_0-1}\sum_{t=0}^{\bar{k}}m_{i,tu+\nu}\lambda_{(e,g)}^{tu+\nu}+\sum_{\nu=u_0}^{u-2}\sum_{t=0}^{\bar{k}-1}
m_{i,tu+\nu}\lambda_{(e,g)}^{tu+\nu}\!+\!\sum_{t=0}^{\bar{d}-1}m_{i,tu+u-1}\lambda_{(e,g)}^{tu+u-1}.
\end{equation*}
Because $\lambda_{(e,g)}=\xi^e\eta^g$ and $\eta$ has multiplicative order $u$, it follows $\lambda_{(e,g)}^{tu}=\xi^{etu}$. By the definition in (\ref{poly}), one can easily verify $f_i(\lambda_{(e,g)})\!=\!\sum_{\nu=0}^{u-1}h_{i,\nu}^{(e)}\lambda_{(e,g)}^\nu\!=\!h_{i}^{(e)}(\lambda_{(e,g)})\;.	$
\end{proof}

\begin{remark}\label{Re2}
Lemma \ref{lemma} shows that for each rack $e\in[0,\bar{n}-1]$, when restricted to the $u$ nodes within rack $e$, the punctured code $C_e$ is actually defined by $\bar{d}$ polynomials of degree at most $u-1$, i.e., $h^{(e)}_i(x)$, $i\in[0,\bar{d}-1]$. If the leading coefficients of the $h^{(e)}_i(x)$'s are already known, then any single node erasure can be recovered from the remaining $u-1$ nodes in the same rack. We call this the local repair property. The idea of reducing the polynomial degree at each local group to ensure the local repair property has been used in constructing optimal locally repairable codes \cite{LRCFamily2014}. Here we further extend the idea to construct linear array codes combining with the product-matrix MBR codes \cite{Kumar2011} for storing the leading coefficients in the cross-rack level.
\end{remark}

\begin{lemma}\label{lemma2}
Consider the leading coefficients of the polynomials $h^{(e)}_i(x)$'s defined in (\ref{poly}). For $e\in[0,\bar{n}-1]$ denote $${\bm h}_e=(h_{0,u-1}^{(e)},h_{1,u-1}^{(e)},...,h_{\bar{d}-1,u-1}^{(e)})^\tau\in F^{\bar{d}}\;.$$
Then $({\bm h}_{0},{\bm h}_{1},...,{\bm h}_{\bar{n}-1})=M_1\Phi$ where \begin{equation}\label{eqphi}\Phi=\begin{pmatrix}1&1&\cdots&1\\1&\xi^{u}&\cdots&\xi^{(\bar{n}-1)u}\\\vdots&\vdots&\vdots&\vdots\\1&(\xi^{u})^{\bar{d}-1}&\cdots&(\xi^{(\bar{n}-1)u})^{\bar{d}-1}
\end{pmatrix}.\end{equation}
Therefore, $\{({\bm h}_{0},{\bm h}_{1},...,{\bm h}_{\bar{n}-1}):\mbox{all possible $M_1$'s as in (\ref{ms2})}\}$ forms an $(\bar{n},\bar{k},\bar{d})$ MBR code.
\end{lemma}

\begin{proof}
The proof is directly from the expression in (\ref{poly}) and the product-matrix construction of MBR codes in \cite{Kumar2011}.
\end{proof}

\begin{theorem}\label{thm2}
Given the code $C$ constructed in (\ref{cons}), any single node erasure can be recovered from all the remaining $u-1$ nodes in the same rack, as well as $\bar{d}$ helper racks each transmitting $\beta=1$ symbol.

In other words, any column of $C$ indexed by $(e^*,g^*)$ can be recovered from the $u-1$ columns indexed by $\{(e^*,g):0\leq g\leq u-1,g\neq g^*\}$ and $\bar{d}$ symbols each of which is a linear combination of the entries of the punctured code $C_e$ for $\bar{d}$ helper racks $e\neq e^*$.
\end{theorem}

\begin{proof}
For any node erasure $(e^*,g^*)$ and any $\bar{d}$ helper racks $e_1,...,e_{\bar{d}}\in[0,\bar{n}-1]-\{e^*\}$, the MBR code proved in Lemma \ref{lemma2} implies that ${\bm h}_{e^*}$ can be recovered from $\bar{d}$ symbols ${\bm \lambda}_{e^*}^\tau{\bm h}_{e_i}, 1\leq i\leq \bar{d}$, where ${\bm \lambda}_{e^*}^\tau=(1,\xi^{e^*u},(\xi^{e^*u})^2,...,(\xi^{e^*u})^{\bar{d}-1})$.
Then by Lemma \ref{lemma}, ${\bm h}_{e_i}$ is a linear combination of the columns of $C_{e_i}$ for $1\leq i\leq \bar{d}$, and the erased column $C_{(e^*,g^*)}$ is a linear combination of ${\bm h}_{e^*}$ and the other $u-1$ columns of $C_{e^*}$. Thus the theorem follows.
\end{proof}

\section{Systematic MBRR codes}
For an $(n,k)$ regenerating code, if there exist $k$ nodes that store all data symbols in uncoded form, the code is called systematic and the $k$ nodes are called systematic nodes. In this section, we provide a transformation to convert the MBRR code constructed in (\ref{cons}) into a systematic code.

Without loss of generality, let the first $k$ nodes be the systematic nodes. Suppose $k=\bar{k}u+u_0$ where $0\leq u_0<u$, so the systematic nodes are all the nodes from rack $0$ to rack $\bar{k}-1$ plus $u_0$ nodes in rack $\bar{k}$. Denote the $B$ data symbols as $s_1,s_2,...,s_B$. Next we are to define the MBRR encoding map that maps $(s_1,...,s_B)$ to a $\bar{d}\times n$ code matrix $C$ such that the first $k$ columns of $C$ contain $s_1,...,s_B$. For simplicity, let $C_{[k]}$ denote the code matrix $C$ restricted to the first $k$ columns. Note from (\ref{eqScalarMBRR}) we know $B=k\bar{d}-\frac{\bar{k}(\bar{k}-1)}{2}$ which means $C_{[k]}$ contain $\frac{\bar{k}(\bar{k}-1)}{2}$ redundant symbols besides the $B$ data symbols. The idea is to determine the redundant symbols from the $B$ data symbols first and then recover the message matrix $\tilde{M}$ by Theorem \ref{thm1} such that $\tilde{M}\Lambda=C$. Thus the systematic encoding map is a composition of $(s_1,...,s_B)\rightarrow \tilde{M}$ and $\tilde{M}\Lambda$. Since $\tilde{M}$ has the same structure as displayed in Fig \ref{fg1}, the resulting code is still an MBRR code. The details are given below.

First we place the $B$ data symbols properly into $C_{[k]}$ except $\frac{\bar{k}(\bar{k}-1)}{2}$ entries which are for the redundant symbols. Specifically, label the columns of $C$ by $(e,g)\in[0,\bar{n}-1]\times[0,u-1]$ and rows by $i\in[0,\bar{d}-1]$, then the column indexed by $(e,u-1)$ for $e\in[0,\bar{k}-2]$ has redundant symbols in its $i$th row as $i\in[e+1,\bar{k}-1]$. The remaining positions are filled up with the data symbols in order. We illustrate the placement in Fig \ref{fg2}.
\begin{figure}[H]
\begin{center}
\includegraphics[width=0.65\textwidth]{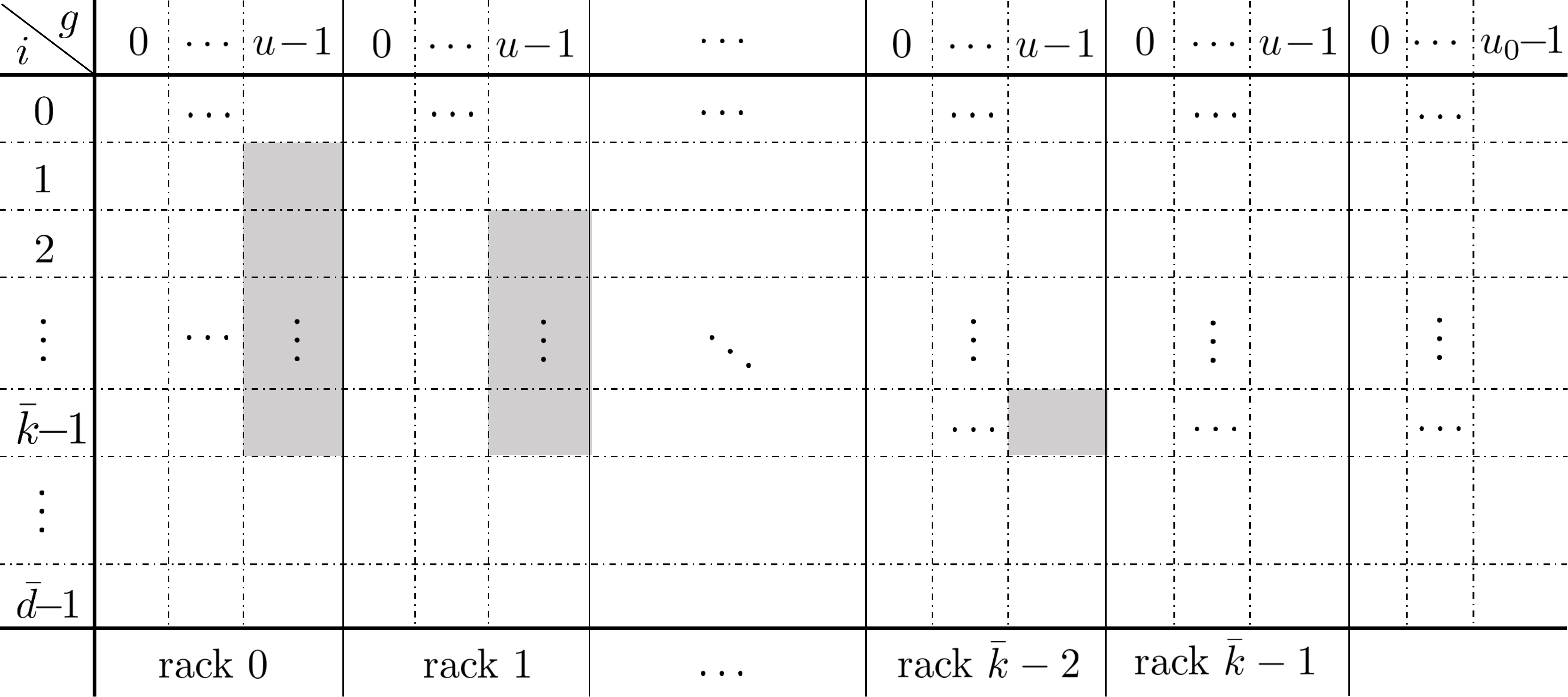}
\end{center}
\caption{An illustration of $C_{[k]}$. The shadowed positions are redundant symbols and the remaining positions are filled up with $B$ data symbols in order. }\label{fg2}
\end{figure}

Next we show the first $k$ columns of $C$ excluding the undetermined redundant symbols are sufficient to recover a message matrix $\tilde{M}$ such that $\tilde{M}$ has the same structure as displayed in Fig \ref{fg1} and $\tilde{M}\Lambda=C$. The key step is to recover the symmetric matrix $\tilde{M}_1$. By Lemma \ref{lemma} we know the code symbols in each row within each rack actually coincides with a local polynomial of degree at most $u-1$. Here we denote the local polynomial by $\tilde{h}_i^{(e)}(x)$. From Fig.\ref{fg2} one can see that for $e\in[0,\bar{k}-1]$, rack $e$ has no redundant symbols in its $i$th row for $i\in[0,e]\cup[\bar{k},\bar{d}-1]$, namely, these rows are already known from the data symbols.  As a result, one can interpolate the local polynomials $\tilde{h}_i^{(e)}(x)$ and obtain the leading coefficients $\tilde{h}_{i,u-1}^{(e)}$. These recovered $\tilde{h}_{i,u-1}^{(e)}$'s are listed in the right side of (\ref{eq9}). Moreover, suppose $$\tilde{M}_1=\begin{pmatrix}
  \tilde{S}& \tilde{T}\\\tilde{T}^\tau&0
\end{pmatrix}\;,$$ where $\tilde{S}$ is a $\bar{k}\times\bar{k}$ symmetric matrix and $\tilde{T}$ is a $\bar{k}\times(\bar{d}-\bar{k})$ matrix. Then by Lemma \ref{lemma2} it has
\begin{equation}
\begin{pmatrix}
  \tilde{S}& \tilde{T}\\\tilde{T}^\tau&0
\end{pmatrix}\begin{pmatrix}1&1&\cdots&1\\1&\xi^{u}&\cdots&\xi^{(\bar{k}-1)u}\\\vdots&\vdots&\vdots&\vdots\\1&(\xi^{u})^{\bar{d}-1}&\cdots&(\xi^{(\bar{k}-1)u})^{\bar{d}-1}
\end{pmatrix}=
\begin{pmatrix}
~\tilde{h}_{0,u-1}^{(0)}~&~\tilde{h}_{0,u-1}^{(1)}~&~\cdots~&~\tilde{h}_{0,u-1}^{(\bar{k}-2)}~&~\tilde{h}_{0,u-1}^{(\bar{k}-1)}~\\
*&\tilde{h}_{1,u-1}^{(1)}&\cdots&\tilde{h}_{1,u-1}^{(\bar{k}-2)}&\tilde{h}_{1,u-1}^{(\bar{k}-1)}\\
*&*&\cdots&\ddots&\vdots\\ *&*&\cdots&*&\tilde{h}_{\bar{k}-1,u-1}^{(\bar{k}-1)}\\
~\tilde{h}_{\bar{k},u-1}^{(0)}~&~\tilde{h}_{\bar{k},u-1}^{(1)}~&~\cdots~&~\tilde{h}_{\bar{k},u-1}^{(\bar{k}-2)}~&~\tilde{h}_{\bar{k},u-1}^{(\bar{k}-1)}~\\
\vdots&\vdots&\vdots&\vdots&\vdots\\
~\tilde{h}_{\bar{d}-1,u-1}^{(0)}~&~\tilde{h}_{\bar{d}-1,u-1}^{(1)}~&~\cdots~&~\tilde{h}_{\bar{d}-1,u-1}^{(\bar{k}-2)}~&~\tilde{h}_{\bar{d}-1,u-1}^{(\bar{k}-1)}~
\end{pmatrix}\;,\label{eq9}\end{equation}
where on the right side of (\ref{eq9}) only the leading coefficients that can be derived from the data symbols by now are written out and entries in the $*$ positions are viewed as unknowns. However, it is enough to recover $\tilde{T}$ and $\tilde{S}$ from the currently known leading coefficients. First, from the last $\bar{d}-\bar{k}$ rows in (\ref{eq9}) one can recover $\tilde{T}$ by multiplying the inverse of a $\bar{k}\times\bar{k}$ Vandermonde matrix. Then substituting the recovered entries of $\tilde{T}$ into the first row in (\ref{eq9}), one can obtain a linear system of equations of the first row entries of $\tilde{S}$. The coefficient matrix is again a $\bar{k}\times\bar{k}$ Vandermonde matrix. Thus one can recover the first row of $\tilde{S}$. Then go to the second row of (\ref{eq9}). Since $\tilde{S}$ is symmetric and its first row has been recovered, there are only $\bar{k}-1$ unknowns in the second row of $\tilde{S}$. Accordingly, the known entries in the second row of the right side of (\ref{eq9}) are enough to recover these unknowns. Continue this process and one can finally recover $\tilde{S}$ row by row. Thus we have proved the following theorem.
\begin{theorem}
The matrix $\tilde{M}_1$ can be uniquely determined by $\{\tilde{h}_{i,u-1}^{(e)}\mid e\in[0,\bar{k}-1], i\in[0,e]\cup[\bar{k},\bar{d}-1]\}$. Furthermore, each entry of $\tilde{M}_1$ can be expressed as a linear combination of the $B$ data symbols.	
\end{theorem}

After recovery of $\tilde{M}_1$, we can fill up the matrix on the right side of (\ref{eq9}). Thus for rack $e\in[0,\bar{k}-1]$ and for each row $i\in[0,\bar{d}-1]$, we have obtained the leading coefficients of the local polynomials $h^{(e)}_i(x)$. Since all data symbols already cover values of each of these polynomials at $u-1$ points, one can easily derive the values at the $u$th points provided the leading coefficients are known. Thus all entries of $C_{[k]}$ have been recovered. Then by the data reconstruction property proved in Theorem \ref{thm1}, one can derive from $C_{[k]}$ the desired message matrix $\tilde{M}$. Through the process, we know each entry of $\tilde{M}$ can be expressed as a linear combination of the $B$ data symbols which defines a preprocess before the product-matrix encoding map in Section \ref{sec1} and finally leads to a systematic MBRR code.


\section{Conclusions}
We explicitly construct regenerating codes that achieve the minimum cross-rack repair bandwidth in the rack aware storage model, i.e., MBRR codes. In general, our code is a specially designed polynomial code that combines the product-matrix MBR code construction for building linear array codes and the multiplicative subgroup structure for reducing the polynomial degree within each local group (i.e., rack). The construction framework developed here is helpful for constructing scalar RRCs (i.e., $\beta=1$) which is appealing in practice due to small sub-packetization. An interesting future work might be constructing MSRR codes along the product-matrix framework.

\vspace{0.3cm}


\end{document}